\documentclass[letterpaper, 10 pt, conference]{ieeeconf}
\IEEEoverridecommandlockouts
\usepackage{enumerate}
\usepackage{times} 
\usepackage{latexsym}
\usepackage{xspace}
\usepackage{amssymb,amsfonts,color}
\usepackage[cmex10]{amsmath}
\usepackage{mathrsfs}
\usepackage{cite}
\usepackage{subfigure}
\usepackage{epsfig} 

\ifx\pdfoutput\undefined
  \usepackage{graphicx}
\else
  \usepackage{graphicx}
\fi

\ifx\pdfoutput\undefined
  \usepackage[hypertex]{hyperref}
\else
  \usepackage[hypertex]{hyperref}
\fi
\usepackage{url}

\newtheorem{definition}{Definition}
\newtheorem{example}{Example}
\newtheorem{lemma}{Lemma}

\newtheorem{remark}{Remark}
\newtheorem{theorem}{Theorem}

\def\tr{^{\mathsf{T}}}

\title{\LARGE \bf
Resilient Continuous-Time Consensus in Fractional Robust Networks
}
\author{Heath~J.~LeBlanc,$^{1}$ Haotian~Zhang,$^{2}$ Shreyas~Sundaram,$^{2}$ and Xenofon~Koutsoukos$^{3}$
\thanks{$^{1}$H. J. LeBlanc is with the Department of Electrical \& Computer Engineering and Computer Science,
        Ohio Northern University, Ada, OH, USA {\tt\small h-leblanc@onu.edu}}%
\thanks{$^{2}$H. Zhang and S. Sundaram are with the Department of Electrical and Computer Engineering at the University of Waterloo, Waterloo, Ontario, Canada {\tt\small {h223zhan,ssundara}@uwaterloo.ca}.}%
\thanks{$^{3}$X. Koutsoukos is with the Department of Electrical Engineering and Computer Science,
        Vanderbilt University, Nashville, TN, USA {\tt\small xenofon.koutsoukos@vanderbilt.edu}}%
}

\begin{document}

\maketitle
\thispagestyle{empty}
\pagestyle{empty}

\begin{abstract}
In this paper, we study the continuous-time consensus problem in the presence of adversaries. 
The networked multi-agent system is modeled as a switched system, where the normal agents have integrator dynamics and the switching signal determines the topology of the network. We consider several models of omniscient adversaries under the assumption that at most a fraction of any normal agent's neighbors may be adversaries. Under this fractional assumption on the interaction between normal and adversary agents, we show that a novel graph theoretic metric, called {\it fractional robustness}, is useful for analyzing the network topologies under which the normal agents achieve consensus.
\end{abstract}

\section{\uppercase{Introduction}}
\label{S:introduction}
Large-scale networks are ubiquitous in nature (e.g., flocks of birds or schools of fish) and are becoming increasingly more pervasive in engineered systems (e.g., large-scale sensor networks). For these systems, reaching consensus in a distributed manner is fundamental to coordination and is therefore a common objective in applications ranging from clock synchronization in sensor networks~\cite{Li06_GlobalClockSynchInSNs} to flocking~\cite{Jadbabaie03_CoordOfGroupsOfMobileAutAgentsUsingNearNeighRule}. However, large-scale distributed systems are susceptible to malicious attacks and failures.  If a security breach occurs, many consensus algorithms fail to achieve consensus, and are therefore not resilient~\cite{Gupta06_OnRobustnessOfDistAlgs}.  

Fault-tolerant and resilient consensus algorithms have been studied extensively over the years~\cite{lynch96:dist_algs, Hromkovic05_DissOfInfoInCommNets}, particularly in the presence of Byzantine nodes under the assumption that at most $F$ of the nodes are compromised~\cite{Lamport82_ByzGeneralProb}. Byzantine nodes are deceptive, can behave arbitrarily within the limitations set by the model of computation, and may be viewed as adversarial in nature. Many of the resilient consensus algorithms studied in the literature are computationally expensive and require at least some global information. However, a class of computationally efficient resilient consensus algorithms that use only local information are developed in \cite{Dolev86_ReachingAppAgreeInPresOfFaults} and \cite{Azad94}, referred to as the {\it Mean-Subsequence-Reduced (MSR)} algorithms~\cite{Azad94}.
The idea behind the MSR algorithms is simple: under the assumption that at most $F$ nodes fail, each normal node removes the largest and smallest $F$ values (i.e., the extreme values) in its neighborhood and takes the average from a subset of the remaining values. MSR algorithms have been used extensively to achieve fault tolerant and resilient consensus (e.g., in clock synchronization~\cite{Li06_GlobalClockSynchInSNs} and robot gathering~\cite{Agmon06_FaultTolGathAlgs4AutMobRobots}). However, the network topological condition for characterizing convergence has long been an open problem.

Recently, it has been shown that traditional graph theoretic metrics (such as connectivity) are inadequate for characterizing the conditions under which the MSR algorithms achieve resilient consensus~\cite{Zhang2012_RobInfoDiffAlgs2LocBdAdv, LeBlanc_LowCompResConsAdv_HSCC12}. Because of the removal of extreme values in MSR algorithms, a property that encapsulates the notion of sufficient {\it local} redundancy of incoming information is needed. This idea is captured by {\it network robustness}~\cite{Zhang2012_RobInfoDiffAlgs2LocBdAdv, LeBlanc2012_ConsOfMANsInPresOfAdvUsingOnlyLocalInfo} and a similar property studied in \cite{Vaidya2012_ItAppByzConsInArbDigraphs}. Equipped with these properties, the necessary and sufficient conditions for convergence of a class of MSR algorithms have been given for the Byzantine model~\cite{Vaidya2012_ItAppByzConsInArbDigraphs} and for a local broadcast version of the Byzantine model~\cite{LeBlanc2012_ConsOfMANsInPresOfAdvUsingOnlyLocalInfo}  (referred to as the malicious adversary), under the assumption that at most $F$ nodes are compromised.

In this paper, we continue our study of continuous-time versions of MSR algorithms~\cite{LeBlanc11_ConsNetMASWithAdv, LeBlanc_LowCompResConsAdv_HSCC12, LeBlanc2012_ConsOfMANsInPresOfAdvUsingOnlyLocalInfo}. We study both malicious and Byzantine adversaries, along with the {\it crash adversary}, which is inspired by the crash faulty robots in robot gathering~\cite{Agmon06_FaultTolGathAlgs4AutMobRobots} and is similar to the fault attack model of \cite{Zeng2012_SecureDistCtrlInUnreliableDNCS}. Instead of assuming an absolute bound on the number of compromised nodes in a normal node's neighborhood, we assume at most a fraction of nodes $f$ may be compromised in the neighborhood. This fractional assumption accounts for the degree of influence on the nodes in the network. We adapt the continuous-time MSR algorithm, referred to as the {\it Adversarial Robust Consensus Protocol (ARC-P)}~\cite{LeBlanc11_ConsNetMASWithAdv, LeBlanc_LowCompResConsAdv_HSCC12, LeBlanc2012_ConsOfMANsInPresOfAdvUsingOnlyLocalInfo}, to the fractional assumption. To analyze ARC-P under the fractional adversary assumption, we consider a fractional form of robustness, introduced in \cite{LeBlanc2012_RACinRobustNets}, and prove separate necessary and sufficient conditions under each of the adversary models. The necessary condition is stated in terms of time-invariant network topology, whereas the sufficient condition applies to time-varying network topologies that are sufficiently robust under a dwell time assumption. The main contribution of this paper is the continuous-time analysis under the fractional scope of threat assumption with the Byzantine, malicious, and crash adversary models.


\section{\uppercase{System Model and Problem Formulation}}
\label{S:Model_Problem}
Consider a time-varying network modeled by the \emph{digraph}  $\mathcal{D}(t)=(\mathcal{V},\mathcal{E}(t))$, where $\mathcal{V}=\{1,...,n\}$ is the \emph{node (agent) set} and $\mathcal{E}(t) \subset \mathcal{V} \times \mathcal{V}$ is the \emph{directed edge set} at time $t$. 
Without loss of generality, the node set is partitioned into a set of $N$ {\it normal agents} $\mathcal{N}=\{1,2,\dotsc,N\}$ and a set of $M$ {\it adversary agents} $\mathcal{A}=\{N+1,N+2,\dotsc,n\}$, with $M=n-N$. Each directed edge $(j,i)\in \mathcal{E}(t)$ indicates that node $i$ can be influenced by node $j$ at time $t$. In this case, we say that agent $j$ {\it conveys} information to agent $i$. The sets of {\it in-neighbors} and {\it out-neighbors} of node $i$ at time $t$ are defined by $\mathcal{N}^{\text{in}}_i(t)=\{j\in\mathcal{V}\colon (j,i)\in\mathcal{E}(t)\}$ and $\mathcal{N}^{\text{out}}_i(t)=\{j\in\mathcal{V}\colon (i,j)\in\mathcal{E}(t)\}$, respectively. The set of all digraphs on $n$ nodes is denoted by $\Gamma_n=\{\mathcal{D}_1,\dotsc,\mathcal{D}_d\}$.

The time-varying topology of the network is governed by a piecewise constant switching signal $\sigma\colon \mathbb{R}_{\geq0}\rightarrow\{1,\dotsc,d\}$. At each point in time $t$, $\sigma(t)$ dictates the topology of the network, 
and $\sigma$ is continuous from the right everywhere. In order to emphasize the role of the switching signal, we denote $\mathcal{D}_{\sigma(t)}=\mathcal{D}(t)$. Note that time-invariant networks are represented by simply dropping the dependence on time $t$.

The agents share state information with one another according to the topology of the network. Each normal agent's state (or value) at time $t$ is denoted as $x_i(t)\in\mathbb{R}$.  In order to handle deceptive adversaries, we let $x_{(j,i)}(t)$ denote the state of agent $j$ {\it intended} for agent $i$ at time $t$. For consistency of notation, we define $x_{(j,i)}(t)$ for {all} $j,i\in\mathcal{V}$, even if $(j,i)\notin \mathcal{E}(t)$. In the case that $j\in\mathcal{N}$ is normal, we define $x_{(j,i)}(t)\equiv x_{j}(t)$ (and in particular, $x_{(j,j)}(t)\equiv x_{j}(t)$). With this terminology, we denote the collective states of all agents in $\mathcal{N}$, $\mathcal{A}$, and $\mathcal{V}$ intended for agent $i$ by
$$
x_{(\mathcal{N},i)}(t)=[x_{1}(t),\dotsc,x_{N}(t)]\tr\in\mathbb{R}^N,
$$
$$
x_{(\mathcal{A},i)}(t)=[x_{(N+1,i)}(t),\dotsc,x_{(n,i)}(t)]\tr\in\mathbb{R}^M,
$$
and
$$
x_{(\mathcal{V},i)}(t)=[x_{(1,i)}(t),\dotsc,x_{(n,i)}(t)]\tr\in\mathbb{R}^n,
$$
respectively. Since $x_{(\mathcal{N},i)}(t)\equiv x_{(\mathcal{N},j)}(t)$ for all $i,j\in\mathcal{V}$, we unambiguously define $x_{\mathcal{N}}(t)=x_{(\mathcal{N},i)}(t)$ for any $i\in\mathcal{V}$. Finally, we denote the vector containing all adversary states intended for the normal agents by $x_{(\mathcal{A}, \mathcal{N})}(t)=[x_{(\mathcal{A},1)}\tr(t),\dotsc,x_{(\mathcal{A},N)}\tr(t)]\tr\in\mathbb{R}^{MN}$.

\subsection{Normal Agent Dynamics}
Each normal agent $i \in \mathcal{N}$ has scalar state $x_i(t)\in\mathbb{R}$ and integrator dynamics given by $\dot{x}_{i} = u_{i}$, where $u_{i} = f_{i,\sigma(t)}(t,x_{\mathcal{N}},x_{(\mathcal{A},i)})$ is a control input. Because there is no prior knowledge about which agents are adversaries, the control input must treat the state information from neighboring agents in the same manner. The system of normal agents are then defined for $t\in\mathbb{R}_{\geq0}$ by
\begin{equation}
\label{E:CoopAgents}
\dot{x}_{\mathcal{N}}(t) = f_{\sigma(t)}(t,x_{\mathcal{N}},x_{(\mathcal{A},\mathcal{N})}), \enspace x_{\mathcal{N}}(0)\in\mathbb{R}^N, \mathcal{D}_{\sigma(t)}\in\Gamma_n,
\end{equation}
where $f_{\sigma(t)}(\cdot)=[f_{1,\sigma(t)}(\cdot),\dotsc,f_{N,\sigma(t)}(\cdot)]^{\mathsf{T}}$. Note that for existence of solutions on $\mathbb{R}_{\geq0}$, the $f_{i,\sigma(t)}(\cdot)$'s must be bounded and piecewise continuous with respect to the adversaries' trajectories. 
These functions should be designed {\it a priori} so that the normal agents can reach consensus {\it without} prior knowledge about the identities of the adversaries.

\subsection{Adversary Model}
The adversary model studied in this paper has two aspects: the {\it threat model} and the {\it scope of threat assumption}. 

\subsubsection{Threat Model}
The threat model defines the types of behaviors allowed by individual adversary nodes. 
The least general threat is the {\it crash adversary}, which is inspired by the \emph{crash fault} studied in mobile robotics~\cite{Agmon06_FaultTolGathAlgs4AutMobRobots}. 
As a fault model, crash-faulty robots fail by simply stopping. Analogously, a crash adversary behaves normally until it is crashed and once crashed, stops changing its state. The crash adversary determines when the node is crashed, but otherwise cannot modify the state of the compromised agent or the values conveyed to other nodes. Crash adversaries -- like all adversary models studied here -- are assumed to be omniscient (i.e., they know all other states and the full network topology; they are aware of the update rules $f_{i,\sigma(t)}(\cdot)$, $\forall i\in\mathcal{N}$; they are aware of which other agents are adversaries; and they know the plans of the other adversaries\footnote{One may take the viewpoint that a centralized omniscient adversary informs and directs the behavior of the individual adversary agents.}). For this reason, the worst case crash times for the adversaries should be considered. This behavior is summarized in the following definition.

\begin{definition}[Crash Adversary]
\label{D:CrashNode}
An agent $k\in\mathcal{A}$ is a \textbf{crash adversary} (or simply \textbf{crash node}) if it is omniscient, and there exists $t_k\in\mathbb{R}_{\geq0}$ (selected by the adversary), such that
\begin{itemize}
 \item agent $k$ behaves normally before $t=t_k$, according to its prescribed update rule, i.e.,
 \begin{equation*}
 \dot{x}_k=f_{k,\sigma(t)}(t,x_{\mathcal{N}},x_{(\mathcal{A},k)}) \text{ for all } t<t_k;
 \end{equation*}
 \item agent $k$ stops changing its state for all $t\geq t_k$, i.e., $x_k(t)=x_k(t_k)$ for all $t\geq t_k$;
 \item agent $k$ conveys the same state to each out-neighbor, i.e., $x_{(k,i)}\equiv x_{(k,j)}$ for all $i,j \in \mathcal{N}^{\text{out}}_k$.
\end{itemize}
\end{definition}

The crash adversary is similar to the fault attack model described in \cite{Zeng2012_SecureDistCtrlInUnreliableDNCS}. The fault attack assumes that the state of the attacked node remains constant, as with a crashed node; however, the constant value imposed by the attack may be arbitrary instead of being fixed at the state value immediately before the attack.

Conversely, the most general threat studied here is the {\it Byzantine adversary}. The Byzantine adversary is motivated by Byzantine faulty nodes studied in distributed computing~\cite{Lamport82_ByzGeneralProb,lynch96:dist_algs}, communication networks~\cite{Hromkovic05_DissOfInfoInCommNets,Jaggi2007_ResilientNetCodInPresOfByzAdv}, and mobile robotics~\cite{Agmon06_FaultTolGathAlgs4AutMobRobots}. Byzantine nodes may behave arbitrarily (under a continuity constraint), are omniscient, and are capable of duplicity, i.e., the values conveyed to their out-neighbors are not necessarily the same. 

The {\it malicious adversary} is essentially a Byzantine node restricted to a local broadcast model of communication. 
A malicious adversary may behave arbitrarily and is omniscient. However, malicious nodes are incapable of duplicity, i.e., every out-neighbor receives the same information.

Malicious nodes have been studied in the detection and identification of misbehaving nodes in discrete-time linear consensus networks~\cite{Sundaram2011_DistFunctCalcViaLinItStratInPresOfMalAgs, Pasqualetti2011_ConsCompInUnrelNets}. In these works, a malicious node is modeled by introducing a disturbance on its input that allows the malicious node to modify its state arbitrarily. To identify the malicious nodes, normal nodes use nonlocal topological information concerning the (time-invariant) network to `invert' the consensus dynamics of the network.

A technical assumption for malicious and Byzantine agents deals with the continuity of the state trajectories of the adversaries. Technically, piecewise continuity of the trajectories of $x_{(\mathcal{A},\mathcal{N})}(t)$ (combined with certain regularity conditions on $f_{\sigma(t)}(\cdot)$)  is sufficient for existence of solutions to (\ref{E:CoopAgents}). However, the trajectories of the normal agents are continuous; therefore, it is feasible that normal agents could use discontinuities in the state trajectories to detect adversaries. Thus, we restrict the trajectories of Byzantine adversaries to be continuous for all $t$. The behavior of Byzantine and malicious agents are summarized as follows.

\begin{definition}[Byzantine and Malicious Agents]
\label{D:ByzantineMaliciousNode}
An agent $k\in\mathcal{A}$ is a \textbf{Byzantine} or \textbf{malicious adversary} if it is omniscient, and
\begin{itemize}
	\item agent $k$'s state trajectories intended for other nodes, $\{x_{(k,i)}(t)\colon i\in\mathcal{V}\}$, are continuous functions of time on $[0,\infty)$;
	\item Byzantine agent $k$'s state trajectory intended for $i$ may be different than the one intended for $j$, i.e.,  $x_{(k,i)}(t) \neq x_{(k,j)}(t)$ is allowed for some $i,j\in\mathcal{V}$;
	\item malicious agent $k$'s state trajectory intended for $i$ must be the same as the one intended for $j$, i.e., $x_{(k,i)}(t) \equiv x_{(k,j)}(t), \forall i,j\in\mathcal{V}$.
\end{itemize}
\end{definition}

\subsubsection{Scope of Threats}
The scope of threat model defines the topological assumptions placed on the adversaries. To account for varying degrees of different nodes, we introduce a fault model that considers an upper bound on the \emph{fraction} of adversaries in any node's neighborhood. This is called the \emph{$f$-fraction local model}.

\begin{definition}[$f$-Fraction Local Set and Threat Model]
A set $\mathcal{S} \subset \mathcal{V}$ is \textbf{$f$-fraction local} if and only if it contains at most a fraction $f$ of agents in the neighborhood of the other agents for all $t$, i.e., $\rvert \mathcal{N}^{\text{in}}_i(t)\bigcap \mathcal{S}\rvert \le \lfloor f \rvert \mathcal{N}^{\text{in}}_i(t)\rvert \rfloor$, $\forall i\in \mathcal{V}\setminus\mathcal{S}$, $f\in[0,1]$. The \textbf{$f$-fraction local model} refers to the case when the set of adversaries is an $f$-fraction local set. 
\end{definition}

It should be emphasized that in time-varying network topologies, the property defining an $f$-fraction local set must hold for all points in time. The $f$-fraction local model is inspired from ideas pertaining to {\it contagion} in social and economic networks~\cite{Easley10_NetCrowdAndMarket}, where a node accepts some new information (behavior or technology) if more than a certain fraction of its neighbors has adopted it.  A scope of threat model similar to the $f$-fraction local model is proposed in \cite{Li06_GlobalClockSynchInSNs} for hierarchical networks to address the problem of resilient clock synchronization in the presence of Byzantine nodes.

%
\subsection{Resilient Asymptotic Consensus}
The Continuous-Time Resilient Asymptotic Consensus (CTRAC) problem is a continuous-time analogue to the {\it Byzantine approximate agreement problem}~\cite{lynch96:dist_algs, Dolev86_ReachingAppAgreeInPresOfFaults}, and is defined as follows. The quantities $M_{\mathcal{N}}(t)$ and $m_{\mathcal{N}}(t)$ are the {\it maximum} and {\it minimum} values of the normal agents at time $t$, respectively.
\begin{definition}[CTRAC]
\label{D:CTRAC}
The normal agents are said to achieve \textbf{continuous-time resilient asymptotic consensus} (CTRAC) in the presence of adversary agents (given a particular adversary model) if
\begin{enumerate}[(i)]
\item $\exists L\in\mathbb{R}$ such that $\lim_{t\rightarrow\infty}x_i(t)=L$ for all $i\in\mathcal{N}$;
\item $x_i(t)\in\mathcal{I}_0=[m_{\mathcal{N}}(0), M_{\mathcal{N}}(0)]$, $\forall t\in\mathbb{R}_{\geq0}$, $i\in\mathcal{N}$,
\end{enumerate}
for any choice of initial values $x_{\mathcal{N}}(0)\in\mathbb{R}^N$.
\end{definition}

The CTRAC problem is defined by two conditions, agreement and safety, along with the type of adversary considered. Condition $(i)$ is an \textit{agreement condition} that requires the states of the normal agents to converge to a common limit, the consensus value, despite the influence of the adversaries.
The safety condition in $(ii)$ ensures that the value chosen by each normal agent lies within the range of `good' values. This is important in safety critical applications, whenever $\mathcal{I}_0$ is a known safe set. 


\section{\uppercase{Resilient Consensus Algorithm}}
\label{S:Algorithm}
Linear consensus algorithms have been extensively studied in the control community for the last few years~\cite{OlfatiSaber07_ConsensusAndCoopInNetMultiagentSys}. In such strategies, at time $t$, each node senses or receives information from its neighbors, and changes its value according to the Linear Consensus Protocol (LCP):
\begin{equation}
\label{E:LCP}
\dot{x}_i(t)=\sum_{j\in \mathcal{N}^{\text{in}}_i(t)}w_{(j,i)}(t)\left(x_{(j,i)}(t)-x_i(t)\right),
\end{equation}
where $x_{(j,i)}(t)-x_i(t)$ is the relative state of agent $j$ with respect to agent $i$ and $w_{(j,i)}(t)$ is a piecewise continuous weight assigned to the relative state at time $t$.

Different conditions have been reported in the literature to ensure that asymptotic consensus is reached~\cite{Ren05_ConsSeekInMultiAgentSysUnderDynChangIntTop,Moreau04_StabCTDistConsAlgs,Jadbabaie03_CoordOfGroupsOfMobileAutAgentsUsingNearNeighRule}.  It is common to assume that the weights are nonnegative, uniformly bounded, and piecewise continuous. That is, there exist constants $\alpha, \beta \in \mathbb{R}_{>0}$, with $\beta\geq\alpha$, such that the following conditions hold:
\begin{itemize}
\item $w_{(j,i)}(t)=0$ whenever $j\not\in \mathcal{N}^{\text{in}}_i(t), \forall i\in\mathcal{N}$, $t\in\mathbb{R}_{\geq0}$;
\item $\alpha \leq w_{(j,i)}(t)\leq\beta$, $\forall j\in \mathcal{N}^{\text{in}}_i(t), i\in\mathcal{N}, t\in\mathbb{R}_{\geq0}$.
\end{itemize}

One problem with LCP given in (\ref{E:LCP}) is that it is not resilient to misbehaving nodes. In fact, it is shown in \cite{Jadbabaie03_CoordOfGroupsOfMobileAutAgentsUsingNearNeighRule, Gupta06_OnRobustnessOfDistAlgs} that a single `leader' node can cause all agents to reach consensus on an arbitrary value of its choosing simply by holding its value constant.

\subsection{Description of ARC-P with Parameter $f$}
\label{S:ARCP2Description}
The Adversarial Robust Consensus Protocol (ARC-P) with parameter $F\in\mathbb{Z}_{\geq0}$ was introduced in \cite{LeBlanc11_ConsNetMASWithAdv} and extended in \cite{LeBlanc2012_ConsOfMANsInPresOfAdvUsingOnlyLocalInfo} to deal with the $F$-total and $F$-local scope of threat models.  By removing the extreme values with respect to the node's own value (the $F$ largest and $F$ smallest values), ARC-P with parameter $F$ is able to achieve resilient asymptotic consensus~\cite{LeBlanc2012_ConsOfMANsInPresOfAdvUsingOnlyLocalInfo}. Under the $f$-fraction local model, a minor modification to this protocol is needed. In this case, the parameter $f\in[0,1/2]$ determines the fraction of neighboring values to view as extreme. For example, if $f=1/3$, then ARC-P with parameter $f$ removes the largest and smallest one third of the neighboring values.


For describing the algorithm, let $F_i(t)=\lfloor fd_i(t)\rfloor$.  Whenever the normal nodes assume the $f$-fraction local model, at most $\lfloor fd_i(t) \rfloor$ of node $i$'s neighbors may be compromised, and the parameter used is $f$.\footnote{Of course, if the scope of threat model {\it assumed} (i.e., at design time) is not the {\it true} scope of threat, then ARC-P may fail to achieve consensus.} 
The following steps describe ARC-P with parameter $f$.

\begin{enumerate}
\item At time $t$, each normal node $i$ obtains the values of its in-neighbors, and forms a sorted list.
\item If there are less than $F_i(t)$ values strictly larger (smaller) than its own value, $x_i(t)$, then normal node $i$ removes all values that are strictly larger (smaller) than its own. Otherwise, it removes precisely the largest (smallest) $F_i(t)$ values in the sorted list.\footnote{Ties may be broken arbitrarily. However, it is required that the algorithm is able to match the correct weights to the values kept.}
\item Let $\mathcal{R}_i(t)$ denote the set of nodes whose values are removed by normal node $i$ in step 2 at time $t$. Each normal node $i$ applies the update\footnote{Note that if all neighboring values are removed, then $\dot{x}_i(t)=0$. }
\begin{equation}
\dot{x}_i(t)=\hspace{-0.3cm}\sum_{j\in \mathcal{N}^{\text{in}}_i(t)\setminus\mathcal{R}_i(t)}w_{(j,i)}(t)\left(x_{(j,i)}(t)-x_i(t)\right).
\label{E:ARCP2Update}
\end{equation}
\end{enumerate}


The set of nodes removed by normal node $i$, $\mathcal{R}_i(t)$, is possibly time-varying. Thus, even if the underlying network is fixed, ARC-P effectively induces switching behavior, which can be viewed as the linear update of (\ref{E:LCP}) with the rule given in step 2 for state-dependent switching.

\section{\uppercase{Fractional Network Robustness}}
\label{S:FracRobust}
Network robustness captures a notion of local redundancy of information flow in the network that is well suited for scope of threat models with absolute bounds on the number of adversaries in a normal node's neighborhood~\cite{LeBlanc2012_ResCoopCtrlNetMAS}. For the $f$-fraction local model, there is no absolute bound on the number of neighbors that may be adversaries. Rather, it stipulates a bound on the {\it fraction} of neighbors that can be adversaries. Hence, for the $f$-fraction local model, we require a fractional notion of robustness. First, we define a $p$-fraction edge reachable set.

\begin{definition}
Given a nonempty digraph $\mathcal{D}$ and a nonempty subset $\mathcal{S}$ of nodes of $\mathcal{D}$, we say $\mathcal{S}$ is a \textbf{$p$-fraction edge reachable set} if there exists $i\in \mathcal{S}$ such that $|\mathcal{N}^{\text{in}}_i|>0$ and $|\mathcal{N}^{\text{in}}_i\setminus\mathcal{S}|\geq \lceil p |\mathcal{N}^{\text{in}}_i|\rceil$, where $0\leq p\leq1$. If $|\mathcal{N}^{\text{in}}_i\setminus\mathcal{S}|=0$ for all $i\in\mathcal{S}$, then $\mathcal{S}$ is 0-fraction edge reachable.
\end{definition}

A set $\mathcal{S}$ is $p$-fraction edge reachable, for $p>0$, if it contains a non-isolated node $i$ (i.e., $d_i>0$) that has at least $\lceil pd_i \rceil$ neighbors outside of $\mathcal{S}$. The parameter $p$ quantifies the ratio of influence from neighbors outside $\mathcal{S}$ to neighbors inside $\mathcal{S}$ for {\it at least} one node inside $\mathcal{S}$.
Note that the notion of fraction edge reachability is also called {\it cohesiveness} in the contagion literature~\cite{Easley10_NetCrowdAndMarket}.


To illustrate $p$-fraction edge reachability, consider the sets $S_1$, $S_2$, and $S_3$ in Figure~\ref{fig:ReachProps}. Each node in $S_1$ has $3/5$ of its neighbors outside $S_1$; so $S_1$ is $\tfrac{3}{5}$-fraction edge reachable. Node 8 has $5/6$ of its neighbors outside of $S_2$, and node 9 only has $4/5$ of its neighbors outside of $S_2$. Thus, $S_2$ is $\tfrac{5}{6}$-fraction edge reachable. Lastly, $S_3$ is a non-isolated singleton, so $S_3$ is 1-fraction edge reachable.

\begin{figure}
\centering
\includegraphics[width=4.8cm]{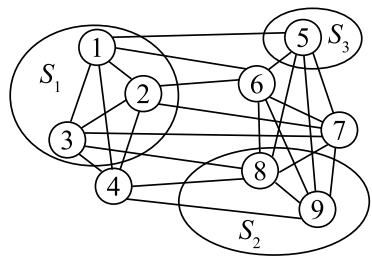} 
\caption{Graph for illustrating edge reachability properties.}
\label{fig:ReachProps}
\end{figure}

The $p$-fraction edge reachability property is defined with respect to a specific set of nodes. We now use this concept to define a network-wide property as follows.

\begin{definition}[$p$-fraction robustness]
A nonempty, nontrivial digraph $\mathcal{D}=(\mathcal{V},\mathcal{E})$ is \textbf{$p$-fraction robust}, with $0\leq p \leq 1$, if for every pair of nonempty, disjoint subsets of $\mathcal{V}$, at least one of the subsets is $p$-fraction edge reachable. If $\mathcal{D}$ is empty or trivial, then $\mathcal{D}$ is 0-fraction robust.
\label{def:p_frac_robust}
\end{definition}

\section{\uppercase{Resilient  Consensus  Analysis}}
\label{S:Results}
We demonstrate in this section how fractional robustness is a useful property for analyzing ARC-P with parameter $f$ under the $f$-fraction local model. More specifically, we show that $f$-fraction robustness is necessary in the presence of crash adversaries and $2f$-fraction robustness is sufficient in the presence of Byzantine adversaries. First, we consider the safety condition. 

\begin{lemma}
\label{L:BoundAveragePhi}
Consider the normal agent $i \in \mathcal{N}$ using ARC-P with parameter $f\in[0,1/2]$ under the $f'$-fraction local (Byzantine) model, where $f'\leq f$. Define $F_i(t)=\lfloor fd_i(t) \rfloor$. Then, for each $t \in \mathbb{R}_{\geq0}$,
\begin{align*}
B(m_{\mathcal{N}}(t)-x_i(t)) &\leq f_{i,\sigma(t)}(t,x_{\mathcal{N}}, x_{(\mathcal{A},i)})\\
&\leq B(M_{\mathcal{N}}(t)-x_i(t)),
\end{align*}
where $B=\beta(n-1-\min_{i\in\mathcal{N},t\geq 0}\{F_i(t)\})$, which implies that the safety condition of the CTRAC problem is ensured.
\end{lemma}
\begin{proof}
If no neighboring values are used, or all values used are equal to $x_i(t)$ at time $t$, then $f_{i,\sigma(t)}(t,x_{\mathcal{N}}, x_{(\mathcal{A},i)})=0$, and the inequality holds. Therefore, assume at least one value not equal to $x_i(t)$ is used in the update at time $t$, say $x_{(j,i)}(t)$. Suppose $x_{(j,i)}(t)>M_{\mathcal{N}}(t)$. Then, by definition $j$ must be an adversary and $x_{(j,i)}(t)>x_i(t)$. Since $i$ uses $x_{(j,i)}(t)$ at time $t$, there must be at least $F_i(t)$ more agents in the neighborhood of $i$ with values at least as large as $x_{(j,i)}(t)$. Hence, these agents must also be adversaries, which contradicts the assumption of at most $F_i(t)$ adversary agents in the neighborhood of $i$ at time $t$. Thus, $x_{(j,i)}(t)\leq M_{\mathcal{N}}(t)$. Similarly, we can show that $x_{(j,i)}(t)\geq m_{\mathcal{N}}(t)$. Since there are at most $n-1$ neighbors of $i$, at least $F_i(t)$ values are removed or equal to $x_i(t)$ (since $d_i(t)>F_i(t)$), and $w_{(j,i)}(t)\leq\beta$ for all $j\in\mathcal{N}^{\text{in}}_i(t)$, it follows that
\begin{align*}
B(m_{\mathcal{N}}(t)-x_i(t))&\leq \hspace{-0.3cm}\sum_{j\in\mathcal{N}^{\text{in}}_i(t)\setminus\mathcal{R}_i(t)}\hspace{-0.3cm}w_{(j,i)}(t)(x_{(j,i)}(t)-x_i(t))\\
&\leq B(M_{\mathcal{N}}(t)-x_i(t)).
\end{align*}
Finally, the fact that any solution of (\ref{E:CoopAgents}) using (\ref{E:ARCP2Update}) is continuous combined with the above inequality implies the result.
\end{proof}

%

It follows from Lemma~\ref{L:BoundAveragePhi} that $M_{\mathcal{N}}(\cdot)$ is nonincreasing and $m_{\mathcal{N}}(\cdot)$ is nondecreasing, respectively. Therefore, if agreement is achieved among the normal agents, then the values of the normal agents must converge to a common limit. For this reason, we focus on proving that the Lyapunov candidate $\Psi(t)=M_{\mathcal{N}}(t)-m_{\mathcal{N}}(t)$ vanishes asymptotically (i.e., agreement is achieved). We show that this Lyapunov function decreases over sufficiently large time intervals whenever the normal nodes update their values according to ARC-P, provided the network is sufficiently robust. Before that, we first provide a necessary condition for the crash model in time-invariant networks.
\subsection{Necessary Condition}
\label{S:ARCP2NecessaryConditions}

\begin{theorem}
\label{T:ARCP2NecessaryfFrac}
Consider a time-invariant network modeled by digraph $\mathcal{D} = (\mathcal{V},\mathcal{E})$ where each normal node updates its value according to ARC-P with parameter $f\in[0,1/2]$ under the $f$-fraction local crash model. If CTRAC is achieved, then $\mathcal{D}$ is $f$-fraction robust.
\end{theorem}
\begin{proof}
Suppose that $\mathcal{D}$ is not $f$-fraction robust. Then, there exist nonempty, disjoint subsets $\mathcal{S}_1, \mathcal{S}_2 \subset \mathcal{V}$ such that neither $\mathcal{S}_1$ nor $\mathcal{S}_2$ is $f$-fraction edge reachable. This means that $|\mathcal{N}_i^{\text{in}}\setminus\mathcal{S}_k|\leq \lfloor fd_i \rfloor$ for all $i\in\mathcal{S}_k$, $k\in\{1,2\}$. Suppose the initial value of each node in $\mathcal{S}_1$ is $a$ and each node in $\mathcal{S}_2$ is $b$, with $a<b$. Let all other nodes have initial values taken from the interval $[a,b]$. Assume all crashing nodes behave normally all the time. 
Then, using ARC-P with parameter $f$, each node $i$ in $\mathcal{S}_1$ removes the $\lfloor fd_i \rfloor$ (or fewer) values greater than $a$ from outside $\mathcal{S}_1$. Likewise, each node $j$ in $\mathcal{S}_2$ removes the $\lfloor fd_j \rfloor$ (or fewer) values less than $b$ from outside $\mathcal{S}_2$. Therefore, each node in $\mathcal{S}_1$ keeps the value $a$ and each node in $\mathcal{S}_2$ keeps the value $b$ for all $t\geq 0$.
\end{proof}

\subsection{Sufficient Condition}
\label{S:ARCP2FractLocMal}

We proved in Theorem~\ref{T:ARCP2NecessaryfFrac} that $f$-fraction robustness is a necessary condition for ARC-P with parameter $f$ to achieve CTRAC in time-invariant networks under the $f$-fraction local crash model (and therefore necessary also for the malicious and Byzantine models). We now show that $p$-fraction robustness, with $p>2f$, is sufficient for the $f$-fraction local Byzantine model (and therefore, also sufficient for the malicious and crash models).

\begin{theorem}
\label{T:ARCP2fFracLocalSuff}
Consider a time-invariant network modeled by digraph $\mathcal{D} = (\mathcal{V},\mathcal{E})$ under the $f$-fraction local Byzantine model. Suppose each normal node updates its value according to ARC-P with parameter $f\in[0,1/2)$.  Then, CTRAC is achieved if the network topology is $p$-fraction robust, where $2f<p\leq 1$.
\end{theorem}
\begin{proof}
We know from Lemma~\ref{L:BoundAveragePhi} that both $M_{\mathcal{N}}(\cdot)$ and $m_{\mathcal{N}}(\cdot)$ are monotone and bounded functions of $t$. Therefore each of them has a limit, denoted by $A_M$ and $A_m$, respectively.   Note that if $A_M=A_m$, then CTRAC is achieved.  We prove by contradiction that this must be the case. The main idea behind the proof is to use the gap between $A_M$ and $A_m$ and combine this with a careful selection of subsets of nodes to show that $\Psi(t)$ will shrink to be smaller than the gap $A_M-A_m$ in finite time (a contradiction).  To this end, suppose that $A_M \ne A_m$ (note that $A_M > A_m$ by definition). Since $M_{\mathcal{N}}(t)\rightarrow A_M$ monotonically, we have $M_{\mathcal{N}}(t)\geq A_M$ for all $t\geq 0$. Similarly, $m_{\mathcal{N}}(t)\leq A_m$ for all $t\geq0$. Moreover, for each $\epsilon>0$ there exists $t_\epsilon>0$ such that $M_{\mathcal{N}}(t) <  A_M + \epsilon$ and $m_{\mathcal{N}}(t) >  A_m - \epsilon$, $\forall t\ge t_{\epsilon}$.  Define constant $\epsilon_0 = (A_M-A_m)/4> 0$.

Next, we define the sets of nodes that are vital to the proof. For any $t_0\geq0$, $t\geq t_0$, $\Delta> 0$, and $\eta>0$, define
\begin{align*}
\mathcal{X}_{M}&(t,t_0,\Delta,\eta)\\
&=\{i\in\mathcal{N}\colon \exists t' \in [t,t+\Delta] \text{ s.t. } x_i(t')>M_{\mathcal{N}}(t_0)-\eta \}
\end{align*}
and
\begin{align*}
\mathcal{X}_{m}&(t,t_0,\Delta,\eta)\\
&=\{i\in\mathcal{N}\colon \exists t' \in [t,t+\Delta] \text{ s.t. } x_i(t')<m_{\mathcal{N}}(t_0)+\eta \}.
\end{align*}

We now proceed by showing that if we choose $\eta$ and $\Delta$ small enough,
then no {\it normal} node can be in both $\mathcal{X}_{M}(t,t_0,\Delta,\eta)$ and $\mathcal{X}_{m}(t,t_0,\Delta,\eta)$ for any $t_0\geq0$ and $t\geq t_0$. First, we require some generic bounds on the normal node trajectories. For $i\in\mathcal{N}$,
we know from Lemma~\ref{L:BoundAveragePhi} that for $\tau\in[t',t]$,
\begin{align*}
\dot{x}_i(\tau)&= \hspace{-0.2cm} \sum_{j\in \mathcal{N}^{\text{in}}_i\setminus\mathcal{R}_i(\tau)} \hspace{-0.1cm}w_{(j,i)}(\tau)\left(x_{(j,i)}(\tau)-x_i(\tau)\right) \\
&\leq B(M_{\mathcal{N}}(t')-x_i(\tau)),
\end{align*}
whenever the derivative exists,\footnote{The solutions of the normal nodes' trajectories are understood in the sense of Carath\'{e}odory. Hence, it is possible that the derivative of the solution does not exist on a set of points in time of Lebesgue measure zero.} where $B=\beta(n-1-\min_{i\in\mathcal{N}, t\geq 0}\{F_i(t)\})$. Using the integrating factor $e^{B(\tau-t')}$, and integrating in the sense of Lebesgue over the time interval $[t',t]$, we have
\begin{equation}
\label{E:NormalFTIncBd}
x_i(t) \leq x_i(t')e^{-B(t-t')}+M_{\mathcal{N}}(t')(1-e^{-B(t-t')}), \enspace \forall t\geq t'.
\end{equation}
By interchanging the roles of $t$ and $t'$, we have
\begin{equation}
\label{E:NormalBTDecBd}
x_i(t) \geq x_i(t')e^{B(t'-t)}+M_{\mathcal{N}}(t)(1-e^{B(t'-t)}), \enspace \forall t\leq t'.
\end{equation}
Similarly, we can show that for $i\in\mathcal{N}$, 
\begin{equation}
\label{E:NormalFTDecBd}
x_i(t) \geq x_i(t')e^{-B(t-t')}+m_{\mathcal{N}}(t')(1-e^{-B(t-t')}), \enspace \forall t\geq t',
\end{equation}
and
\begin{equation}
\label{E:NormalBTIncBd}
x_i(t) \leq x_i(t')e^{B(t'-t)}+m_{\mathcal{N}}(t)(1-e^{B(t'-t)}), \enspace \forall t\leq t'.
\end{equation}
Now fix $\eta\leq\epsilon_0=(A_M-A_m)/4$ and $\Delta<\log(3)/B$, and suppose $i\in\mathcal{X}_{M}(t,t_0,\Delta,\eta)$. Then $\exists t'\in[t,t+\Delta]$ such that $x_i(t')>M_{\mathcal{N}}(t_0)-\eta$. Combining this with (\ref{E:NormalFTDecBd}), it follows that for $s\in[t',t+\Delta]$,
\begin{align}\nonumber
x_i(s)&\geq x_i(t')e^{-B(s-t')}+m_{\mathcal{N}}(t')(1-e^{-B(s-t')})\\ \nonumber
&> (M_{\mathcal{N}}(t_0)-\eta)e^{-B(s-t')}+m_{\mathcal{N}}(t_0)(1-e^{-B(s-t')})\\ \nonumber
&\geq (A_M-\eta)e^{-B(s-t')}+m_{\mathcal{N}}(t_0)-A_m e^{-B(s-t')}\\ \nonumber
&\geq m_{\mathcal{N}}(t_0)+(A_M-A_m)e^{-B(s-t')}\\ \nonumber
&\hspace{1cm}-\frac{A_M-A_m}{4} e^{-B(s-t')}\\ \nonumber
&\geq m_{\mathcal{N}}(t_0)+\frac{3}{4}(A_M-A_m)e^{-B\Delta}\\ \nonumber
&> m_{\mathcal{N}}(t_0)+\frac{A_M-A_m}{4}\\ \nonumber
&\geq m_{\mathcal{N}}(t_0)+\eta, \nonumber
\end{align}
where we have used the fact that $\Delta<\log(3)/B$. Similarly, using (\ref{E:NormalBTDecBd}), it follows that for $s\in[t,t']$,
\begin{align}\nonumber
x_i(s)&\geq x_i(t')e^{B(t'-s)}+M_{\mathcal{N}}(s)(1-e^{B(t'-s)})\\ \nonumber
&> (M_{\mathcal{N}}(t_0)-\eta)e^{B(t'-s)}+M_{\mathcal{N}}(s)(1-e^{B(t'-s)})\\ \nonumber
&\geq M_{\mathcal{N}}(s)-\eta e^{B(t'-s)}\\ \nonumber
&\geq M_{\mathcal{N}}(s) - \frac{A_M-A_m}{4} e^{B\Delta} \\ \nonumber
&> A_M - \frac{3}{4}(A_M-A_m) \\ \nonumber
&\geq A_m + \frac{1}{4}(A_M-A_m) \\ \nonumber
&\geq m_{\mathcal{N}}(t_0)+\eta. \nonumber
\end{align}
Therefore, $i\notin\mathcal{X}_{m}(t,t_0,\Delta,\eta)$.

Similarly, with the given choices for $\eta$ and $\Delta$, if $j\in\mathcal{X}_{m}(t,t_0,\Delta,\eta)$, then $\exists t'\in[t,t+\Delta]$ such that $x_j(t')<m_{\mathcal{N}}(t_0)+\eta$. It follows from (\ref{E:NormalFTIncBd}) that for $s\in[t',t+\Delta]$,
\begin{align}\nonumber
x_j(s)&\leq x_j(t')e^{-B(s-t')}+M_{\mathcal{N}}(t')(1-e^{-B(s-t')})\\ \nonumber
&< (m_{\mathcal{N}}(t_0)+\eta)e^{-B(s-t')}+M_{\mathcal{N}}(t_0)(1-e^{-B(s-t')})\\ \nonumber
&\leq M_{\mathcal{N}}(t_0)+(\eta + m_{\mathcal{N}}(t_0)-M_{\mathcal{N}}(t_0))e^{-B(s-t')}\\ \nonumber
&\leq M_{\mathcal{N}}(t_0)+\left(\frac{A_M-A_m}{4}-(A_M-A_m)\right)e^{-B(s-t')}\\ \nonumber
&\leq M_{\mathcal{N}}(t_0)-\frac{3}{4}(A_M-A_m)e^{-B\Delta}\\ \nonumber
&< M_{\mathcal{N}}(t_0)-\frac{A_M-A_m}{4} \\ \nonumber
&\leq M_{\mathcal{N}}(t_0)-\eta, \nonumber
\end{align}
where we have used the fact that $\Delta<\log(3)/B$. Finally, using (\ref{E:NormalBTIncBd}), it follows that for $s\in[t,t']$,
\begin{align}\nonumber
x_j(s)&\leq x_j(t')e^{B(t'-s)}+m_{\mathcal{N}}(s)(1-e^{B(t'-s)})\\ \nonumber
&< (m_{\mathcal{N}}(t_0)+\eta)e^{B(t'-s)}+m_{\mathcal{N}}(s)(1-e^{B(t'-s)})\\ \nonumber
&\leq m_{\mathcal{N}}(s)+\eta e^{B(t'-s)}\\ \nonumber
&\leq A_m + \frac{A_M-A_m}{4} e^{B\Delta} \\ \nonumber
&< A_m + \frac{3}{4}(A_M - A_m)\\ \nonumber
&\leq M_{\mathcal{N}}(t_0) - \frac{A_M-A_m}{4}\\ \nonumber
&\leq M_{\mathcal{N}}(t_0)-\eta. \nonumber
\end{align}
Thus, $j\notin\mathcal{X}_{M}(t,t_0,\Delta,\eta)$. This shows that $\mathcal{X}_{M}(t,t_0,\Delta,\eta)$ and $\mathcal{X}_{m}(t,t_0,\Delta,\eta)$ are disjoint for appropriate choices of the parameters.

Next, we show that by choosing $\epsilon$ small enough, we can define a sequence of sets,
$$
\mathcal{X}_{M}^k \triangleq \mathcal{X}_{M}(t_\epsilon+k\Delta,t_\epsilon,\Delta,\epsilon_k), \quad k=0,1,\dotsc,N,
$$
and
$$
\mathcal{X}_{m}^k \triangleq \mathcal{X}_{m}(t_\epsilon+k\Delta,t_\epsilon,\Delta,\epsilon_k), \quad k=0,1,\dotsc,N,
$$
where $N=|\mathcal{N}|$, so that we are guaranteed that by the $N$th step, at least one of the sets contains no normal nodes. This will be used to show that $\Psi$ has shrunk below $A_M-A_m$. Toward this end, let $\epsilon_0=(A_M-A_m)/4$ and $\Delta<\log(3)/B$. Then fix
\begin{equation*}
\label{E:EpsDefinedCTFTotal}
\epsilon < \frac{1}{2}\left[\frac{\alpha}{B}(1-e^{-B\Delta})e^{-B\Delta}\right]^{2N}\epsilon_0.
\end{equation*}
For $k=0,1,\dotsc,N$, define $\epsilon_k=[\tfrac{\alpha}{B}(1-e^{-B\Delta})e^{-B\Delta}]^{2k}\epsilon_0$, which results in
$$
\epsilon_0>\epsilon_1>\dots>\epsilon_N>2\epsilon>0.
$$
Observe that by definition, there is at least one normal node in $\mathcal{X}_{M}^0$ and $\mathcal{X}_{m}^0$ (the ones with extreme values), and we have shown above that all of the $\mathcal{X}_{M}^0$ and $\mathcal{X}_{m}^0$ are disjoint.  The $p$-fraction robust assumption (with $p>2f$) ensures that there exists a (normal) node $i$ in either $\mathcal{X}_{M}^0$ or $\mathcal{X}_{m}^0$ with at least $\lceil p d_i \rceil$ neighbors outside of either $\mathcal{X}_{M}^0$ or $\mathcal{X}_{m}^0$, respectively. For node $i$, at most $2\lfloor f d_i \rfloor$ of these values are thrown away (with at most $\lfloor f d_i \rfloor$ of them as adversaries, under the $f$-fraction local model, and at most $\lfloor f d_i \rfloor$ of these strictly smaller, or larger, than node $i$'s value). Because $p>2f$, it follows that $\lceil p d_i \rceil-2\lfloor f d_i \rfloor\geq1$. Therefore, at least one normal value outside of $i$'s set (either $\mathcal{X}_{M}^0$ or $\mathcal{X}_{m}^0$) is used. Assume $i\in\mathcal{X}_M^0$ has at least $\lceil p d_i \rceil$ neighbors outside of its set. Then, at least one of the values from $i$'s neighbors outside of $\mathcal{X}_M^0$ is used for almost all $t\in[t_\epsilon, t_\epsilon+\Delta]$. Then,
\begin{align*}
\dot{x}_i(t)&\leq \alpha(M_{\mathcal{N}}(t_\epsilon)-\epsilon_0-x_i(t))\\
&\hspace{1cm}+(B-\alpha)(M_{\mathcal{N}}(t_\epsilon)-x_i(t))\\
&\leq -B x_i(t)+BM_{\mathcal{N}}(t_\epsilon)-\alpha \epsilon_0,
\end{align*}
for almost all $t\in[t_\epsilon, t_\epsilon+\Delta]$. Using this, we can show
\begin{align}\nonumber
x_i(t_\epsilon+\Delta)&\leq x_i(t_\epsilon)e^{-B\Delta}+(M_{\mathcal{N}}(t_\epsilon)-\tfrac{\alpha \epsilon_0}{B})(1-e^{-B\Delta})\\ \nonumber
&\leq M_{\mathcal{N}}(t_\epsilon)-\tfrac{\alpha}{B}(1-e^{-B\Delta})\epsilon_0.\nonumber
\end{align}
Using this with (\ref{E:NormalFTIncBd}) to bound $x_i(t)$ for $t\in[t_\epsilon+\Delta, t_\epsilon+2\Delta]$, we see that for all $t\in[t_\epsilon+\Delta, t_\epsilon+2\Delta]$,
\begin{align}\nonumber
x_i(t)
&\leq M_{\mathcal{N}}(t_\epsilon)-\tfrac{\alpha}{B}(1-e^{-B\Delta})e^{-B(t-t_\epsilon-\Delta)}\epsilon_0 \\ \nonumber
&\leq M_{\mathcal{N}}(t_\epsilon)-\tfrac{\alpha}{B}(1-e^{-B\Delta})e^{-B\Delta}\epsilon_0 \\ \nonumber
&\leq M_{\mathcal{N}}(t_\epsilon)-\epsilon_1.  \nonumber
\end{align}
Thus, $i\notin \mathcal{X}_M^1$. The next step is to show that $j\notin\mathcal{X}_m^1$ whenever $j$ is a normal node with $j\notin\mathcal{X}_m^0$. Whenever $j\notin\mathcal{X}_m^0$, it means that $x_j(t_\epsilon+\Delta)\geq m_{\mathcal{N}}(t_\epsilon)+\epsilon_0$. Using this with (\ref{E:NormalFTDecBd}) to lower bound $x_j(t)$ for $t\in[t_\epsilon+\Delta, t_\epsilon+2\Delta]$, we see that
$$
x_j(t)\geq m_{\mathcal{N}}(t_\epsilon)+\epsilon_0 e^{-B\Delta} \geq m_{\mathcal{N}}(t_\epsilon)+\epsilon_1.
$$
Hence, $j$ is also not in $\mathcal{X}_m^1$, as claimed.  Likewise, we can show using (\ref{E:NormalFTIncBd}) that $j\notin\mathcal{X}_M^1$ whenever $j$ is a normal node with $j\notin\mathcal{X}_M^0$. Therefore, if $i\in \mathcal{X}_M^0$ uses at least one normal neighbor's value outside of its set, we are guaranteed that $|\mathcal{X}_M^1|<|\mathcal{X}_M^0|$ and $|\mathcal{X}_m^1|\leq|\mathcal{X}_m^0|$. Using a similar argument, we can show that if $i\in \mathcal{X}_m^0$ has at least $\lceil pd_i \rceil$ neighbors outside of its set, we are guaranteed that $|\mathcal{X}_m^1|<|\mathcal{X}_m^0|$ and $|\mathcal{X}_M^1|\leq|\mathcal{X}_M^0|$.

Now, if both $\mathcal{X}_M^1$ and $\mathcal{X}_m^1$ are nonempty, we can repeat the above argument to show that either $|\mathcal{X}_m^2|<|\mathcal{X}_m^1|$ or $|\mathcal{X}_M^2|<|\mathcal{X}_M^1|$, or both. It follows by induction that as long as both $\mathcal{X}_M^j$ and $\mathcal{X}_m^j$ are nonempty, then either $|\mathcal{X}_m^{j+1}|<|\mathcal{X}_m^j|$ or $|\mathcal{X}_M^{j+1}|<|\mathcal{X}_M^j|$ (or both), for $j=1,2,\dots$. Since $|\mathcal{X}_m^0|+|\mathcal{X}_M^0|\leq N$, there exists $T<N$ such that at least one of $\mathcal{X}_M^T$ and $\mathcal{X}_m^T$ is empty.
If $\mathcal{X}_M^T=\emptyset$, then $M_{\mathcal{N}}(t_\epsilon+T\Delta)\leq M_{\mathcal{N}}(t_\epsilon)-\epsilon_T<M_{\mathcal{N}}(t_\epsilon)-2\epsilon$. Similarly, if $\mathcal{X}_m^T=\emptyset$, then $m_{\mathcal{N}}(t_\epsilon+T\Delta)\geq m_{\mathcal{N}}(t_\epsilon)+\epsilon_T>m_{\mathcal{N}}(t_\epsilon)+2\epsilon$. In either case, $\Psi(t_\epsilon+T\Delta)<A_M-A_m$ and we reach the desired contradiction.
\end{proof}

We now extend the above result to time-varying networks.

\begin{theorem}
\label{T:ARCP2fFracLocalTimeVar}
Consider a time-varying network modeled by $\mathcal{D}(t) = (\mathcal{V},\mathcal{E}(t))$ under the $f$-fraction local Byzantine model. Let $\{t_k\}$ denote the switching times of $\sigma(t)$ and assume that $t_{k+1}-t_k\geq \tau$ for all $k$. Suppose each normal node updates its value according to ARC-P with parameter $f\in[0,1/2)$.  Then, CTRAC is achieved if there exists $t' \geq 0$ such that $\mathcal{D}(t)$ is $p$-fraction robust, where $2f<p\le 1$, for all $t\geq t'$.
\end{theorem}
\begin{proof}
The proof follows the contradiction argument of the proof of Theorem~\ref{T:ARCP2fFracLocalSuff}, but here we use the dwell time assumption. In this case, let
$$\Delta<\min\{\log(3)/B, \tfrac{\tau}{N} \}.$$
Fix
\begin{equation*}
\epsilon < \frac{1}{2}\left[\frac{\alpha}{B}(1-e^{-B\Delta})e^{-B\Delta}\right]^{2N}\epsilon_0,
\end{equation*}
and let $t_{\epsilon}'\geq 0$ be a point in time such that $M_{\mathcal{N}}(t)<A_M+\epsilon$ and $m_{\mathcal{N}}(t)>A_m-\epsilon$ for all $t\geq t_{\epsilon}'$. Define $t''=\max\{t',t_{\epsilon}'\}$. Then, associated to the switching signal $\sigma(t)$, we define $t_{\epsilon}$ as the next switching instance after $t''$, or $t''$ itself if there are no switching instances after $t''$. Since $\Delta<\tau/N$, the same sequence of calculations can be used (as in the proof of Theorem~\ref{T:ARCP2fFracLocalSuff}) to show that $\Psi(t_{\epsilon}+T\Delta)<A_M-A_m$.
\end{proof}


\section{\uppercase{Conclusion}}
\label{S:Conclusion}
This paper studies the continuous-time resilient asymptotic consensus problem. The adversary models studied are omniscient and have a scope that is fractional in nature (i.e., at most a fraction $f$ of nodes in any normal nodes neighborhood are assumed to be compromised). Under these assumptions, we show that a fractional version of the Adversarial Robust Consensus Protocol (ARC-P) achieves consensus among the normal nodes if the network is $2f$-fraction robust and only if the network is $f$-fraction robust. Determining a tight condition for these adversary models is a matter of future work.



\end{document}